\documentclass[aps,pra,twocolumn,a4paper,groupedaddress,10pt]{revtex4-1}
\usepackage{amssymb,amsthm,amsmath,amsfonts}
\usepackage{graphicx,ulem,mathptmx}
\usepackage[pdftex,dvipsnames,usenames]{xcolor}
\usepackage[colorlinks=true,urlcolor=blue,citecolor=blue,linkcolor=blue]{hyperref}

\newtheorem{definition}{Definition}
\newtheorem{lemma}{Lemma}
\newtheorem{corollary}{Corollary}

\begin{document}

\title{Quasiprobability representation of quantum coherence}

\author{J. Sperling}\email{jan.sperling@physics.ox.ac.uk}
\affiliation{Clarendon Laboratory, University of Oxford, Parks Road, Oxford OX1 3PU, United Kingdom}

\author{I. A. Walmsley}
\affiliation{Clarendon Laboratory, University of Oxford, Parks Road, Oxford OX1 3PU, United Kingdom}

\date{\today}

\begin{abstract}
	We introduce a general method for the construction of quasiprobability representations for arbitrary notions of quantum coherence.
	Our technique yields a nonnegative probability distribution for the decomposition of any classical state.
	Conversely, quantum phenomena are certified in terms of signed distributions, i.e., quasiprobabilities, and a residual component unaccessible via classical states.
	Our unifying method combines well-established concepts, such as phase-space distributions in quantum optics, with resources of quantumness relevant for quantum technologies.
	We apply our approach to analyze various forms of quantum coherence in different physical systems.
	Moreover, our framework renders it possible to uncover complex quantum correlations between systems, for example, via quasiprobability representations of multipartite entanglement.
\end{abstract}

\maketitle

\section{Introduction}\label{Sec:Introduction}

	The discovery of quantum mechanics led to paradigms describing aspects of nature beyond classical physics.
	For instance, the fact that a quantum state is represented by complex probability amplitudes provided an explanation of surprising observations, such as interference phenomena of particles.
	However, a demand to unify classical concepts with the ideas from quantum theory remained.
	For this reason, Wigner formulated a phase-space representation to provide a classical picture of the wave function \cite{W32}.
	This approach enables us to understand to what extent quantum phenomena are compatible or incompatible with their classical counterparts.
	While such a construct yields a nonnegative distribution in classical statistical mechanics, it can include negative quasiprobabilities when applied to quantum systems.

	In general, the technique of quasiprobability distributions presents a remarkably successful approach in modern quantum physics.
	For example, quasiprobabilities can be applied to confirm fundamental predictions of quantum physics, such as nonlocality \cite{LJJ09,RLHL13}, macrorealism \cite{JRL17}, and contextuality \cite{S08}.
	Moreover, negative quasiprobabilities serve as a resource for quantum information protocols \cite{VFGE12,SLR17}.

	The visualization of quantum features through quasiprobabilities is an efficient and intuitive approach to characterize quantum systems.
	Consequently, the phase-space representation became an essential tool for field theories, such as quantum optics \cite{MW95}.
	In this case, the Glauber-Sudarshan representation \cite{G63,S63} is a quasiprobability distribution which allows one to expand a state of light in terms of coherent states, which resemble the behavior of classical harmonic oscillators \cite{S26}.
	The Glauber-Sudarshan representation also inspired the formulation of a manifold of generalized quasiprobability distributions \cite{CG69,AW70}, which include the Wigner function and other well-known phase-space distributions \cite{H40,K65} as special cases.
	Beyond bosonic systems, quasiprobability distributions for fermions have been recently established as well \cite{CD06,DJB16,P16,DJB17,JRD17}.

	Connections between quasiprobability distributions and entanglement led to insights into the nature of quantum correlations \cite{DMWS06,WFPT17}, which includes the study of hybrid systems combining continuous-variable and discrete-variable degrees of freedom \cite{WMV97,ASCBZV17}.
	Also, quantum properties can be inferred from the conditional \cite{ASCBZV17,M07} and marginal \cite{KVCBAP12,PLLSZZZKN17} distributions of quasiprobabilities.
	Furthermore, even processes can be characterized using quasiprobability distributions \cite{RKVGZB13}, and multi-time phase-space distributions provide access to temporal quantum correlations \cite{V08,KVS17}.

	Quasiprobabilities are experimentally accessible \cite{SBRF93,LMKMIW96,BRWK99}.
	This includes the application to pulsed light \cite{LCGS10} and large atomic ensembles \cite{MZJHCV15}.
	The reconstruction of phase-space distributions is even possible with imperfect detection schemes \cite{HSRHMSS16,BTBSSV18} and relates to the dual operators of a positive-operator-valued measure \cite{KSVS18}.
	Quasiprobabilities also serve as the foundation to experimentally analyze other quantum features of light, such as the orbital angular momentum \cite{BQTSLSKB15}.
	In this case, however, classical angular momentum states and, similarly, classical atomic states differ from those introduced for the harmonic oscillator \cite{AD71,ACGT72}.
	Still, the formulation of quasiprobability distributions is possible using the Poincar\'e sphere as the underlying phase space \cite{A81,A93,DAS94,A99}.
	Applications can be found, for example, in the fields of quantum metrology \cite{WBIMH92} and the entanglement characterization \cite{LLA09}.
	A related question addresses nonclassical states of the polarization of light, which can be similarly represented by negative quasiprobabilities in theory and experiment \cite{KLLRS02,L06,MSPGRHKLMS12,SBKSL12,KRW14}.

	One challenge in continuous-variable systems is that quasiprobabilities can become highly singular \cite{C65,S16}.
	For this reason, regularization methods have been introduced \cite{K66,KV10} and generalized to multipartite and time-dependent systems \cite{ASV13,KVS17}.
	This allows for the direct experimental sampling of regular negative quasiprobabilities of quantum light \cite{KVHS11,ASVKMH15}.
	Notably, such irregularity problems do not occur in discrete-variable systems.

	The investigation of quasiprobabilities in finite dimensions, such as  Wigner-function-based approaches \cite{LP98,KMR06,PB11} and their generalizations (see, e.g., \cite{RMG05}), renders it possible to obtain a deeper insight into nonclassical properties specific to discrete-variable systems.
	For example, quasiprobabilities can be generalized to qudits \cite{FE09}, which are particularly relevant in the context of quantum information and communication \cite{NC00}, as well as to study the quantum property of the spin \cite{P96,TMS01,P12}.
	Furthermore, finite-dimensional quasiprobabilities relate to so-called weak measurements \cite{FLT17} and enable the efficient estimation of actual measurement-outcome probabilities \cite{PWB15}.

	For different physical systems, the notion of what properties a classical state has to have can alter \cite{P72}, leading to different types of quasiprobabilities \cite{BM98}.
	This means the analysis of quantum phenomena depends on the particular choice of classical reference states \cite{WB12}.
	Yet, some quantum correlations between subsystems, such as entanglement, are not affected by such a local basis dependence \cite{SAWV17,SLR17a}.
	To provide a unified framework for the study of different quantum phenomena, the concept of quantum coherence has been introduced (see Ref. \cite{SAP17} for an overview).
	The theory of quantum coherence yields a quantitative description of quantum effects \cite{BCP14,LM14} in which the quantum-physical superposition principle plays a pronounced role \cite{SV15}.
	Further, the operational meaning of this quantum resource for quantum protocols and processes has been investigated \cite{WY16,YMGGV16}.
	Recent applications also performed a comparative analysis of classical and nonclassical evolutions \cite{SW18} and studied the impact of the spin statistics \cite{SPBW17}.
	Moreover, if one considers bipartite product states as the classical basis, one obtains all separable states via statistical mixtures of classical states \cite{W89}.
	Consequently, quantum coherence includes entanglement, which can be also characterized in terms of negative quasiprobabilities \cite{STV98,SV09quasi}.
	Thus, the concept of quantum coherence renders it possible to study various quantum phenomena in a common framework.

	Today, the description of quantum phenomena through negative quasiprobabilities continues to have a broad impact on modern research and is the basis for a profound understanding of quantum physics.
	While remarkable progress has been made recently in the quasiprobability-based description of quantum coherence (see, e.g., \cite{PJSZ15,Z16,TESMN16,DF17}), the actual construction of an optimal decomposition in terms of quasiprobability distributions over classical states remained an open problem.

	In this paper we derive a method for the construction of quasiprobability distributions for general forms of quantum coherence.
	Since such a decomposition is, in general, not unique, we derive a method which is optimal in the sense that it necessarily yields a nonnegative decomposition for all classical states.
	By contrast, quantum coherence is determined via negativities within the quasiprobability distribution over the set of classical states or the inability to expand the state under study in terms of any linear combination of classical states, leading to an orthogonal residual component.
	To apply our technique, we analytically study quantum coherences for different scenarios relevant in quantum information and quantum physics.
	Furthermore, we uncover quantum correlations between multiple degrees of freedom, including multipartite entanglement.
	Thus, we provide a construction method for obtaining optimal quasiprobabilities to investigate phenomena of quantum coherence in a unified manner.

	The paper is organized as follows.
	We explain the need for optimal decompositions and define essential concepts in Sec. \ref{Sec:Prelim}.
	Section \ref{Sec:Decomp} includes the mathematical framework which results in the systematic formulation of our technique.
	The method is then applied to different notions of quantum coherence in Sec. \ref{Sec:Apps}, followed by a study of diverse notions of quantum correlations in Sec. \ref{Sec:Corr}.
	Our work is summarized in Sec. \ref{Sec:Conclusion}.

\section{Preliminaries}\label{Sec:Prelim}

\subsection{Density operator decomposition}\label{Subsec:Density}

	A fundamental paradigm of quantum physics is that the state of a system can be described in terms of a density operator.
	More specifically, any quantum state can be written as statistical mixture of pure ones,
	\begin{align}
		\label{eq:GeneralDecomposition}
		\hat\rho=\int dP_\mathrm{cl.}(\psi)|\psi\rangle\langle\psi|.
	\end{align}
	Here $P_\mathrm{cl.}$ is a classical, i.e., normalized and nonnegative, probability distribution over the set of pure states.
	The density operator is a nonnegative and Hermitian operator with a unit trace norm, $\hat\rho=\hat\rho^\dag\geq0$ and $\mathrm{tr}(\hat\rho)=1$.

	There are two problems with the decomposition of mixed states in terms of pure ones.
	First, the decomposition \eqref{eq:GeneralDecomposition} is ambiguous.
	For example, the density operator $\hat\rho'=(|0\rangle\langle 0|+2|1\rangle\langle 1|)/3,$ in which $\{|0\rangle,|1\rangle\}$ denotes an orthonormal basis, can be also expanded as $\hat\rho'=(4|\phi\rangle\langle\phi|+5|\chi_{-}\rangle\langle\chi_{-}|)/9$, where $|\phi\rangle=(|0\rangle+|1\rangle)/\sqrt{2}$ and $|\chi_{\pm}\rangle=(|0\rangle\pm2|1\rangle)/\sqrt{5}$.
	Second, a quasiprobability, i.e., $P\neq P_\mathrm{cl.}$, also enables a decomposition of the state.
	For our example, we can find the quasi-mixture representation $\hat\rho'=(4|1\rangle\langle 1|+4|\phi\rangle\langle\phi|-5|\chi_{+}\rangle\langle \chi_{+}|)/3$ with a negative ``probability'' of the contribution $\chi_{+}$.

	To overcome these deficiencies, one can use the spectral decomposition for an optimal state expansion,
	\begin{align}
		\label{eq:SpectralDecomposition}
		\hat\rho=\sum_n p_n|\psi_n\rangle\langle\psi_n|,
	\end{align}
	which is based on the eigenvalue problem of the density operator $\hat\rho|\psi_n\rangle=p_n|\psi_n\rangle$.
	The eigenvalues in this decomposition are nonnegative and normalized, $p_n\geq 0$ for all $n$ and $\sum_n p_n=1$.
	Using the spectral decomposition, we get a valid classical and unique (up to degenerate eigenvalues) probability distribution for Eq. \eqref{eq:GeneralDecomposition}.
	Specifically, the probability density is given by $dP_\mathrm{cl.}(\psi)/d\psi=\sum_{n} p_n \delta(\psi-\psi_n)$, where $\delta$ denotes the Dirac delta distribution.
	For our considerations, it is also worth mentioning that an unphysical density operator, i.e., $\hat\rho''\ngeq 0$, necessarily includes at least one negative eigenvalue.
	In fact, the spectral decomposition \eqref{eq:SpectralDecomposition} proves that this operator $\hat\rho''$ cannot be a convex combination \eqref{eq:GeneralDecomposition} of pure physical states.

	Thus, the spectral decomposition yields an optimal decomposition in terms of pure states and enables us distinguish valid density operators from operators outside the convex set of physical states.
	We are going to generalize this concept for identifying classical and nonclassical states.

\subsection{Quasiprobabilities for harmonic oscillators}\label{Subsec:GSDecomposition}

	Sometimes a restriction to specific families of pure states for the decomposition \eqref{eq:GeneralDecomposition} is useful for the aim of studying specific quantum phenomena.
	This implies that the spectral decomposition \eqref{eq:SpectralDecomposition} cannot be applied under this restriction, because the eigenvectors are not necessarily in the family of states under study.

	A prominent example for such a scenario is the Glauber-Sudarshan representation of quantum states of light \cite{G63,S63},
	\begin{align}
		\label{eq:GSDecomposition}
		\hat\rho=\int dP(\alpha)|\alpha\rangle\langle\alpha|,
	\end{align}
	which exclusively employs coherent states $|\alpha\rangle$ since they resemble most closely the behavior of a classical harmonic oscillator \cite{S26}.
	Note that in the context of modern resource theories, the commonly used name ``coherent state'' for a classical optical state might be confusing, but it represents the historic approach to define this state as a superposition of photon-number eigenstates.
	In order to expand any quantum state in the form of Eq. \eqref{eq:GSDecomposition}, it is unavoidable to allow for quasiprobabilities, $P\neq P_\mathrm{cl.}$.
	In fact, this turns out to be beneficial as the impossibility to find a classical Glauber-Sudarshan distribution defines nonclassical light and separates quantum-optical phenomena from classical statistical optics \cite{TG65,M86}.
	An example of an experimentally reconstructed, negative Glauber-Sudarshan quasiprobability density can be found in Ref. \cite{KVPZB08} and is similar to the one presented in Fig. \ref{fig:SPATS} for a state studied in Ref. \cite{AT92}.
	It is also worth mentioning that the class of Glauber-Sudarshan distributions can be ambiguous too \cite{S16}.

\begin{figure}
	\includegraphics[width=\columnwidth]{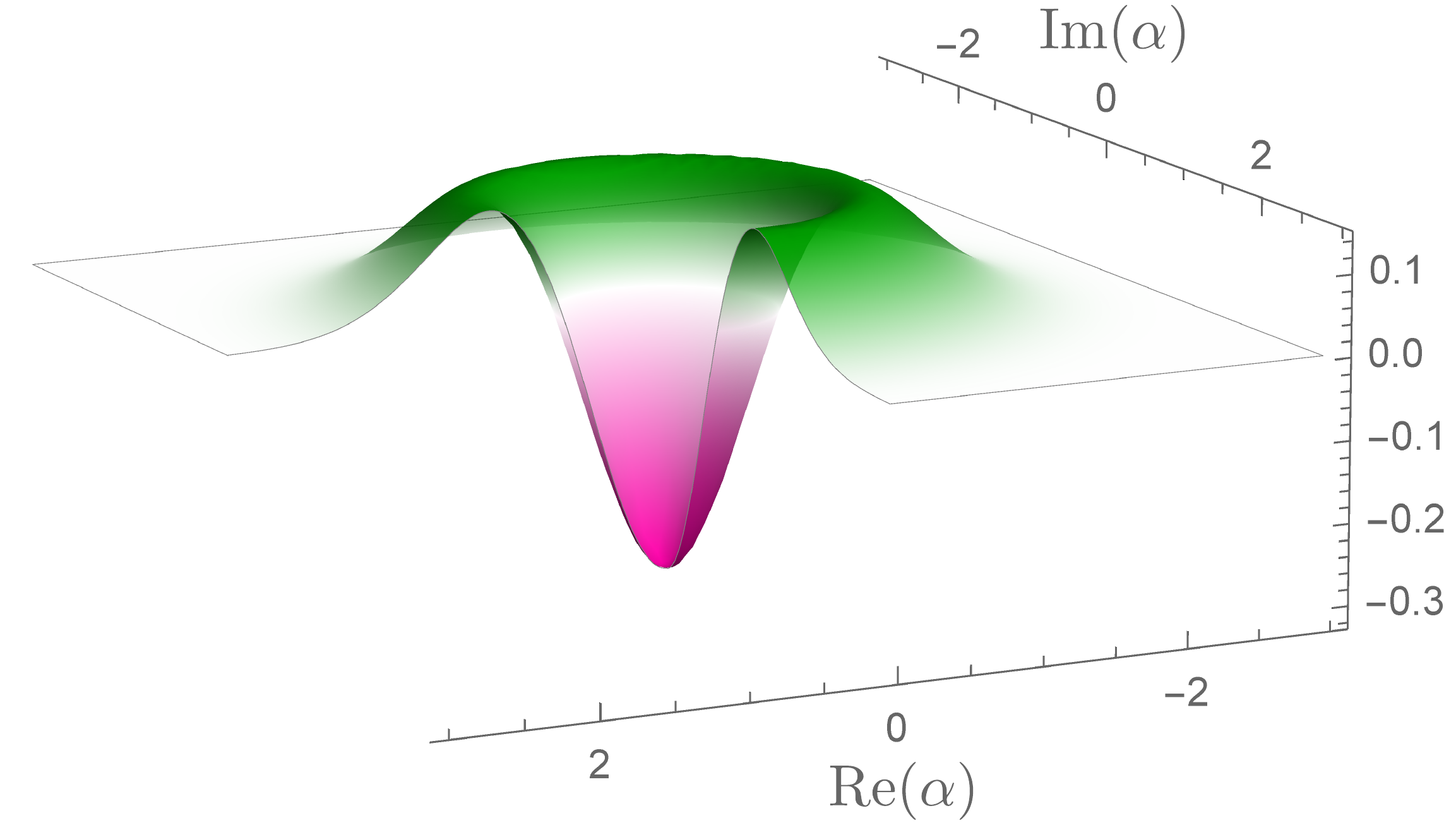}
	\caption{
		Quasiprobability density of a single-photon-added thermal state, $dP(\alpha)/d\alpha=([1+\bar n]|\alpha|^2-\bar n])e^{-|\alpha|^2/\bar n}/(\pi\bar n^3)$ for $\bar n=1$.
		The negative contributions (magenta) certify that this state cannot be written as a convex mixture of coherent states.
	}\label{fig:SPATS}
\end{figure}

	The Glauber-Sudarshan decomposition \eqref{eq:GSDecomposition} is an intuitive approach to visualize nonclassical quantum states of light (see Fig. \ref{fig:SPATS}).
	Therefore, we are going to generalize such a method to other forms of quantum phenomena beyond quantized radiation fields described through harmonic oscillators.

\subsection{Quantum coherence}

	To address general quantum phenomena, the notion of quantum coherence has been introduced, which combines various concepts of nonclassicality.
	A recent comprehensive review can be found in Ref. \cite{SAP17}.

	For defining a general classical state, one supposes that $\mathcal C$ is a closed subset of pure states, $|c\rangle\langle c|$ with $\langle c|c\rangle=1$, which are consistent with a classical model to be specified.
	In the most general case, these states neither are orthogonal nor form a basis.
	A mixed classical state can written as the convex mixture of pure ones,
	\begin{align}
		\label{eq:IncoherentDecomposition}
		\hat\rho=\int dP_\mathrm{cl.}(c) |c\rangle\langle c|,
	\end{align}
	where $P_\mathrm{cl.}$ is a probability distribution over $\mathcal C$.
	Thus, any classical state is in the closure, with respect to the trace norm, of the convex span of pure classical states $\mathrm{conv}(\mathcal C)$.
	Then $\hat\rho$ is a so-called incoherent state because it can be described solely in terms of statistical mixtures of pure classical states and it does not require coherent quantum superpositions.
	Conversely, a state which is not classical exhibits quantum coherences and is consequently referred to as a nonclassical state.

	In this work we restrict ourselves to finite-dimensional systems.
	This prevents us from dealing with the mathematical peculiarities of highly singular distributions, as they are typical for continuous-variable systems \cite{S16}.
	Moreover, this is specifically useful for applications in quantum information technology employing discrete variables.
	Further, this enables us to write the continuous representation in Eq. \eqref{eq:IncoherentDecomposition} as a discrete sum $\hat\rho=\sum_{n} p_n |c_n\rangle\langle c_n|$
	(see Ref. \cite{G07} for an introduction to convex geometry in finite spaces).

	Motivated by the intuitive quasiprobability decomposition for the harmonic oscillator [Eq. \eqref{eq:GSDecomposition} and Fig. \ref{fig:SPATS}] and the uniqueness of the spectral decomposition [Eq. \eqref{eq:SpectralDecomposition}], we are going to devise a technique for the decomposition of quantum states in terms of classical ones.
	There are two important aspects which are satisfied by our approach.
	First, we develop a constructive approach beyond bare existence statements, which are typically addressed.
	Second, our method ensures that any classical (i.e., incoherent) state is optimally described in terms of a classical distribution.
	Thus, if our reconstructed probabilities fails to be a nonnegative representation of the state under study, this state is uniquely certified to be a nonclassical one.

\section{Quantum state decomposition}\label{Sec:Decomp}

\subsection{Initial remarks}

	A trivial observation to be made is that any quantum state $\hat\rho$ can be decomposed as
	\begin{align}
		\label{eq:QuantumDecomposition}
		\hat\rho=\sum_n p_n|c_n\rangle\langle c_n|+\hat\rho_\mathrm{res.},
	\end{align}
	where $\vec p=(p_n)_n$ is an arbitrary (typically, signed) distribution, $|c_n\rangle\langle c_n|\in\mathcal C$, and $\hat\rho_\mathrm{res.}$ is a residual part of the state.
	The residual component is not accessible when restricting to the states in $\mathcal C$ (cf. Fig. \ref{fig:SchemeDecomposition}).
	A state is classical if a representation with $\vec p\geq0$ (i.e., $p_n\geq0$ for all $n$) and $\hat\rho_\mathrm{res.}=0$ exists [cf. Eq. \eqref{eq:IncoherentDecomposition}].
	Conversely, as depicted in Fig. \ref{fig:SchemeDecomposition}, nonclassicality results from the impossibility of a classical representation.
	Thus, nonclassical states require a quasiprobability representation, $\vec p\ngeq0$, or a non-vanishing component $\hat\rho_\mathrm{res.}$.
	
\begin{figure}
	\includegraphics[width=\columnwidth]{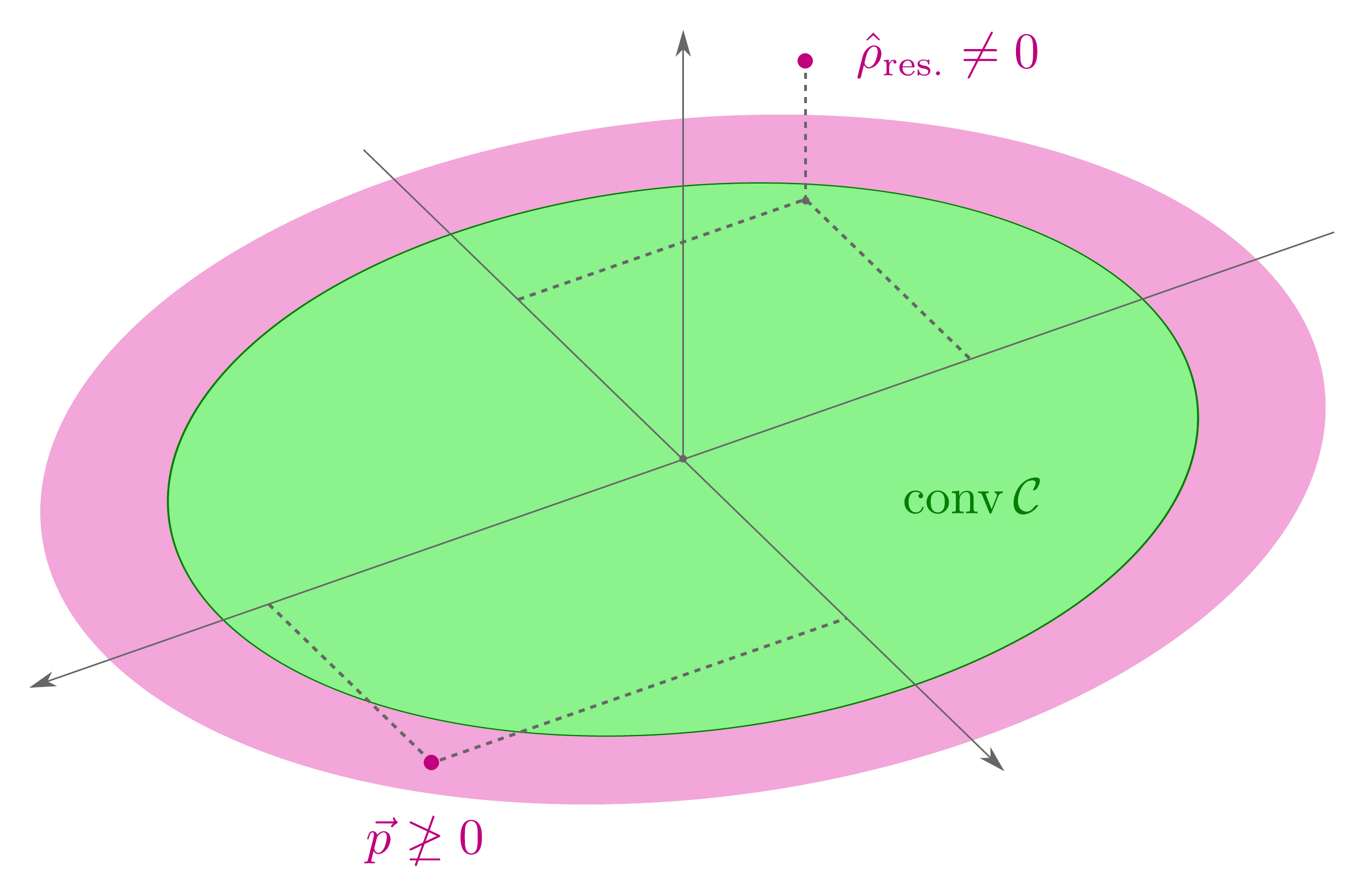}
	\caption{
		Schematic decomposition.
		For classical states $\hat\rho$, i.e., elements of the convex span of $\mathcal C$ (green area), the decomposition \eqref{eq:IncoherentDecomposition} applies.
		Nonclassical states can be further separated into two classes [Eq. \eqref{eq:QuantumDecomposition}]: elements in the space spanned by $\mathcal C$ (magenta area with $\vec p\ngeq0$) and elements out of this plane ($\hat\rho_\mathrm{res.}\neq0$).
		Note that, in general, these classes are not disjoint because we can simultaneously have $\hat\rho\neq 0$ and $\vec p\ngeq 0$.
	}\label{fig:SchemeDecomposition}
\end{figure}

	We emphasize that the notions of classicality studied here are based on the existence of a convex decomposition of a state.
	It does not mean that any decomposition has to be nonnegative, as we pointed out for the decomposition of general density operators (Sec. \ref{Subsec:Density}).
	This means even if, for example, $\vec p\ngeq0$ is true for one decomposition, there might exist another one with ${\vec p}'\geq 0$.
	Characterizing all possible decompositions becomes specifically cumbersome when $\mathcal C$ includes an infinite number of elements.

	Therefore, we aim to construct an optimal convex decomposition of any element $\hat \rho\in\mathrm{conv}(\mathcal C)$.
	Our approach is based on the premature observation that pure states $\mathcal D\subset\mathcal C$ defined by an optimal distance to $\hat\rho$ enable the desired decomposition (see Fig. \ref{fig:SchemeComposition}).
	In this section, we develop the rigorous mathematical framework to prove that this intuitive picture holds true and it additionally yields a constructive approach to obtain a decomposition with $\vec p\geq0$ for any $\hat\rho\in\mathrm{conv}\,\mathcal C$.

\begin{figure}
	\includegraphics[width=\columnwidth]{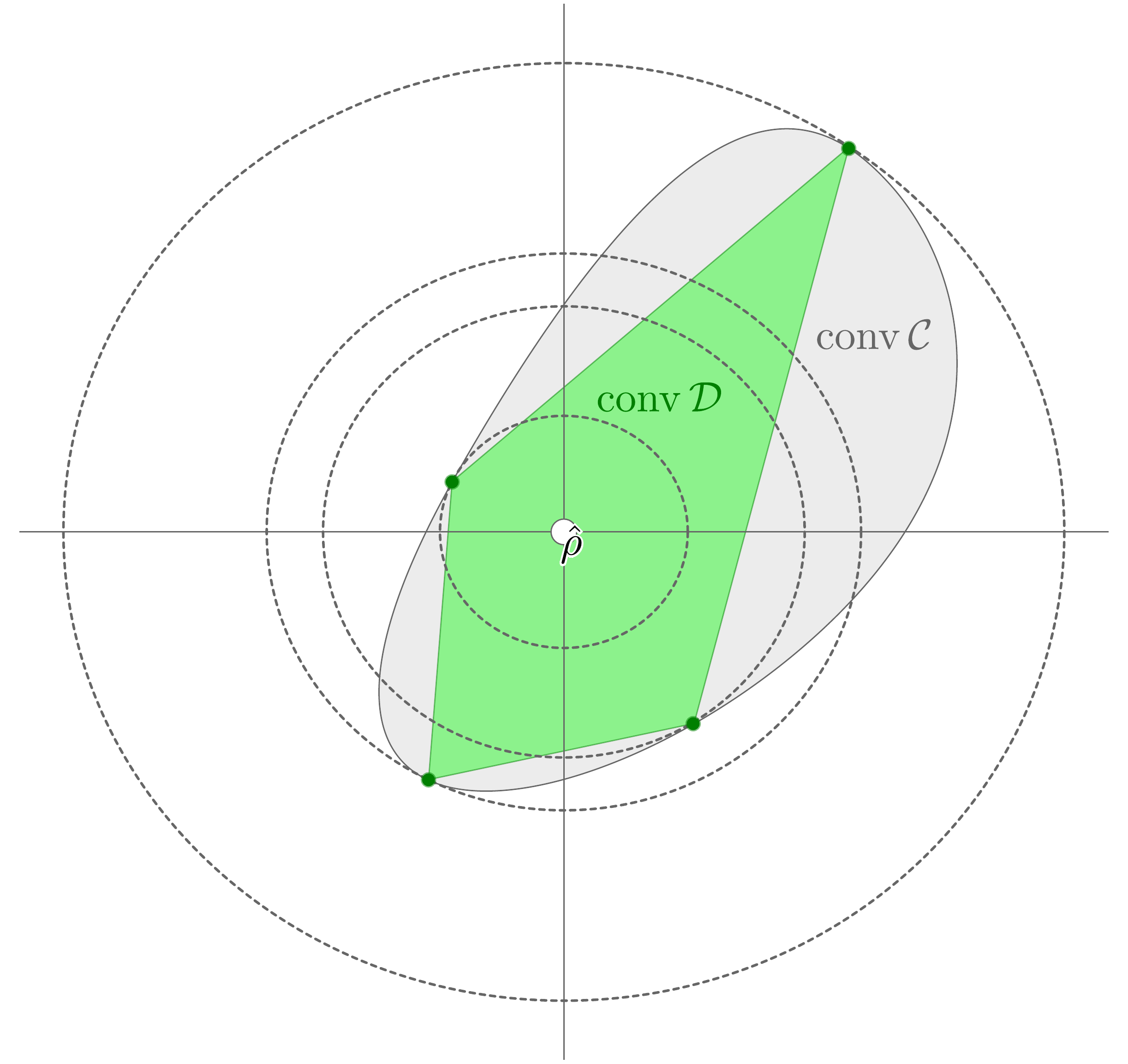}
	\caption{
		Decomposition of $\hat\rho$, which is an element of $\mathrm{conv}\,\mathcal C$ (light gray area).
		The optimal distances (dashed circles) of this state to $\mathcal C$ (here the boundary of gray area) results in a subset of points $\mathcal D\subset\mathcal C$ (green bullets).
		For an optimal decomposition, we are going to show that $\hat\rho\in\mathrm{conv}\,\mathcal D$ (light green polytope) holds true.
	}\label{fig:SchemeComposition}
\end{figure}
	
	Let us briefly establish some notations and methods common in convex geometry \cite{G07} to perform this task.
	Since the case $\hat\rho\in\mathcal C$ is trivial, let us further assume that $\hat\rho$ is not a pure state.
	Furthermore, it is convenient to shift the problem so that the element under study is in the origin.
	This means we translate the space such that the element of interest is $\hat\rho_0=0$.
	This is indicated by the index ``$0$,'' which defines the operation
	\begin{align}
		\hat x\mapsto \hat x_0=\hat x-\hat\rho.
	\end{align}
	This map yields the translated set $\mathcal C_0=\{|c\rangle\langle c|-\hat\rho:|c\rangle\langle c|\in\mathcal C\}$ and also implies $\hat\rho_0=\hat\rho-\hat\rho=0\in\mathrm{conv}(\mathcal C_0)$.

	In addition, we can define a subspace of Hermitian operators, $\mathcal X=\mathrm{lin}(\mathcal C)$.
	Note that $\mathcal X$ is defined only over real numbers $\mathbb R$ as the operator $i\hat x$ is not Hermitian for $\hat x=\hat x^\dag$ ($\hat x\neq0$).
	This space can be equipped with the Hilbert-Schmidt scalar product
	\begin{align}
		\label{eq:HSproduct}
		(\hat x|\hat y)=\mathrm{tr}(\hat x\hat y)
	\end{align}
	for $\hat x,\hat y\in\mathcal X$ which yields the norm $\|\hat x\|=(\hat x|\hat x)^{1/2}$.
	For example, $\|\hat\rho\|^2$ describes the purity of the state.
	The dual space $\mathcal X'$ is the set of all linear maps $f:\mathcal X\to\mathbb R$ and can be represented by
	\begin{align}
		f(\hat x)=(\hat f|\hat x)
	\end{align}
	for an operator $\hat f\in\mathcal X$.
	It is worth mentioning that the dual space plays a fundamental role for the detection of nonclassicality, e.g., in terms of entanglement witnesses \cite{HHH96}.

\subsection{Convex decomposition}

	We need two definitions to formalize our treatment.

	\begin{definition}\label{Def:StatPoint}
		Let $\mathcal B_\varepsilon(\hat x)=\{\hat z\in\mathcal X: \|\hat z-\hat x\|\leq \varepsilon\}$ be the ball of radius $\varepsilon$ centered at $\hat x$.
		An element $\hat y$ of a closed and bounded set $\mathcal M\neq\emptyset$ is a {\em stationary point} to $\hat x$ if
		\begin{align*}
			&\exists \varepsilon>0 \ \forall \hat z\in\mathcal B_\varepsilon(\hat x)\cap\mathcal M: \|\hat y-\hat x\|\geq \|\hat z-\hat x\|\\
		\text{or}\quad
			&\exists \varepsilon>0 \ \forall \hat z\in\mathcal B_\varepsilon(\hat x)\cap\mathcal M: \|\hat y-\hat x\|\leq \|\hat z-\hat x\|.
		\end{align*}
	\end{definition}

	This definition means that stationary points have an optimal (minimal or maximal) distance to $\hat y$ in an $\varepsilon$ neighborhood.
	Because $\mathcal M$ is a closed and bounded subset of a finite-dimensional space (thus, compact), there are least stationary maximum and minimum points, namely, the global maximum and minimum.
	Also note that Definition \ref{Def:StatPoint} implies that an isolated point, an $\hat x\in\mathcal M$ which is the only point in an $\varepsilon$ neighborhood, i.e., $\mathcal B_\varepsilon(\hat x)\cap\mathcal M=\{\hat x\}$, is a stationary point.

	\begin{definition}\label{Def:ExhaustiveSet}
		A closed set $\mathcal N$ is {\em exhaustive} if for all $f\in\mathcal X'$ there exist $\hat x_+,\hat x_-\in\mathcal N$ such that
		\begin{align*}
			f(\hat x_+)\geq 0
			\quad\text{and}\quad
			f(\hat x_-)\leq 0.
		\end{align*}
	\end{definition}

	An exhaustive set ensures that for any hyperplane $\{\hat x\in\mathcal X:f(\hat x)=0\}$ there exist stationary points above (or on) and under (or on) this hyperplane.
	Note that $\hat x_+=\hat x_-$ is possible when $f(\hat x_\pm)=0$.
	Exhaustive sets are important as they ensure the convex decomposition of the zero element.

	\begin{lemma}\label{Lem:ZeroGen}
		$0\in\mathrm{conv}(\mathcal N)$ holds true for an exhaustive set $\mathcal N$.
	\end{lemma}
	\begin{proof}
		Assume $0\notin\mathrm{conv}(\mathcal N)$.
		From the Hahn-Banach separation theorem it follows that there exists $f\in\mathcal X'$ such that for all $\hat y\in\mathrm{conv}(\mathcal N)$, $f(\hat y)<0=f(0)$ holds true.
		This contradicts the property of the exhaustive set $\mathcal N$.
		Thus, the assumption $0\notin\mathrm{conv}(\mathcal N)$ is false.
	\end{proof}

	The combination of both definitions proves that the stationary points form an exhaustive set.

	\begin{lemma}\label{Lem:PMPoints}
		The set of stationary points $\mathcal D_0\subset\mathcal C_0$ to $\hat\rho_0=0$ is exhaustive.
	\end{lemma}
	\begin{proof}
		To construct a contradiction, let us assume that the negation is true, i.e., $\exists f\in\mathcal X'$ such that for all $\hat x$ stationary points of $\mathcal C_0$ to $0$, we have $f(\hat x)<0$.
		Since $0\in\mathrm{conv}(\mathcal C_0)\setminus\mathcal C_0$, there exist $\hat c_j\in\mathcal C_0$ and $p_j>0$ such that $0=\sum_j p_j\hat c_j$.
		Because $f$ is linear and $f(0)=0$, we have $f(\hat y)\geq 0$ for at least one $\hat y=\hat c_{j'}$.
		Because of our assumption, $\hat y$ is not stationary.
		Thus, we define the nonempty closed set $\mathcal M_f=\{\hat y\in\mathcal C_0: f(\hat y)\geq0\}$ which is separated by $f$ from all stationary points.
		However, $\mathcal M_p$ has to include at least one stationary point $\hat x_{+}$, e.g., a global maximum.
		This yields the contradiction; our assumption is not true.
		Analogously, we can show the existence of a stationary point $\hat x_{-}\in\mathcal C_0$ to $0$ for which $f(\hat x_{-})\leq 0$.
	\end{proof}

	With this lemma, we have identified the elements of $\mathcal C_0$ which allow for the desired convex decomposition.
	Namely, $\hat\rho_0\in\mathrm{conv}\,\mathcal C_0$ can be written as a convex combination of its stationary points in $\mathcal C_0$.
	It is also worth emphasizing that a subset of stationary points can be sufficient as long as it remains to be exhaustive.

	Undoing the initial shift $\hat\rho_0\mapsto\hat\rho$, we can directly conclude the following corollary, which also confirms our intuition presented in Fig. \ref{fig:SchemeComposition}.

	\begin{corollary}\label{Theo:ConvDec}
		Any $\hat \rho\in\mathrm{conv}(\mathcal C)$ can be written as a convex combination of stationary points in $\mathcal C$ and likewise any subset $\mathcal D$ thereof for which $\mathcal D_0$ is exhaustive.
	\end{corollary}

	As a consequence of this corollary, the element $\hat\rho$ can be written as
	\begin{align}
		\hat \rho=\sum_{\hat x\in\mathcal D} p_x \hat x.
	\end{align}
	Performing a projection onto $\hat y\in\mathcal D$, we find
	\begin{align}
		\forall \hat y\in\mathcal D:  (\hat y|\hat\rho)=\sum_{\hat x\in\mathcal D} p_x ( \hat y | \hat x ).
	\end{align}
	This defines a linear equation which has, according to Corollary \ref{Theo:ConvDec}, at least one solution with nonnegative probabilities, $p_x\geq0$ for all $x$.
	Therefore, we can directly conclude the following.

	\begin{corollary}\label{Theo:RecCO}
		We define $\vec g=[(\hat y|\hat\rho)]_{\hat y\in\mathcal D}$, $\vec p=[p_x]_{\hat x\in\mathcal D}$, and $G=[( \hat y | \hat x )]_{\hat x,\hat y\in\mathcal D}$.
		The linear system
		\begin{align*}
			G\vec p=\vec g
		\end{align*}
		has a solution $\vec p\geq 0$ (i.e., $p_x\geq0$ for all $\hat x$) and $\hat \rho=\sum_{\hat x\in\mathcal D} p_x \hat x$.
	\end{corollary}

	Note that $G$ is the symmetric, positive-semidefinite Gram-Schmidt matrix.
	A solution of the linear problem can be obtained, e.g., with the simplex algorithm.
	In the case in which the kernel of $G$ is empty, $G$ is invertible and $\vec p$ is unique, $\vec p=G^{-1}\vec g$.

\subsection{Concluding remarks}

	After this rigorous treatment let us conclude this section with some practical observations.
	Let us describe how to compute the stationary points and obtain the probability distribution.
	For this reason, it is worth noting that the function $\hat x\mapsto(\hat x|\hat x)=\|\hat x\|^2$ [cf. Eq. \eqref{eq:HSproduct}] is differentiable.
	Further, let us assume that elements of $\mathcal C$ are given through the differentiable function $t\mapsto\hat \Gamma(t)\in\mathcal C$, where $t$ represents a (possibly piecewise and multidimensional) parametrization of those elements.
	From Definition \ref{Def:StatPoint} it follows that we can obtain the stationary points from the optimization of the function
	\begin{align}
		\|\hat \rho-\hat\Gamma(t)\|^2=(\hat\rho|\hat\rho)-2(\hat\rho|\hat\Gamma(t))+(\hat\Gamma(t)|\hat\Gamma(t)).
	\end{align}
	This means the derivative $\partial_t$ of this expression has to be zero.
	Classical pure states yield $(\hat\Gamma(t)|\hat\Gamma(t))=\mathrm{const.}$ because of $\mathrm{tr}(|c\rangle\langle c|c\rangle\langle c|)=1$.
	Thus, $\partial_t \|\hat \rho-\hat \Gamma(t)\|^2=0$ is equivalent to
	\begin{align}
		\label{eq:OptProbl}
		\partial_t (\hat \rho|\hat \Gamma(t))=0.
	\end{align}
	The resulting solutions for all parametrizations, in addition to isolated points, yields all stationary points, defining the set $\mathcal D$.

	So far we focused on the construction of a convex decomposition of a classical state $\hat\rho\in\mathrm{conv}(\mathcal C)$.
	Now let us consider nonclassical states $\hat\rho\notin\mathrm{conv}(\mathcal C)$.
	In this case, we can also compute stationary points.
	However, we get $\vec p\ngeq 0$ or $\hat\rho_\mathrm{res.}\neq0$ (cf. Fig. \ref{fig:SchemeDecomposition}).
	In the former case, this means that the linear equation in Corollary \ref{Theo:RecCO} has no positive solution.
	In the latter case, the reconstructed state does not coincide with the actual state.
	That is, the residual component, given by
	\begin{align}
		\label{eq:RC}
		\hat \rho_\mathrm{res.}=\hat \rho-\sum_{|c\rangle\langle c|\in\mathcal D} p_c |c\rangle\langle c|,
	\end{align}
	does not vanish.

	Therefore, the derived approach yields a constructive as well as necessary and sufficient criterion for identifying quantum coherences.
	After this general treatment, let us apply our method to various examples for characterizing quantum coherence and quantum correlations in different systems.

\section{Applications}\label{Sec:Apps}

\subsection{Spectral decomposition}\label{Subsec:Spectral}

	As a consistency check, let us assume that $\mathcal C$ contains all pure states.
	Thus, we can define a parametrization which maps any vector to a pure state, $t=\langle \psi|\mapsto\hat\Gamma(t)=|\psi\rangle\langle\psi|/\langle \psi|\psi\rangle$.
	Consequently, using the definition \eqref{eq:HSproduct} of the Hilbert-Schmidt scalar product, the optimization problem in Eq. \eqref{eq:OptProbl} reads
	\begin{align}
		0=\partial_{\langle\psi|}\frac{\langle\psi|\hat\rho|\psi\rangle}{\langle\psi|\psi\rangle}=-\frac{\hat\rho|\psi\rangle}{\langle\psi|\psi\rangle}+\frac{\langle \psi|\hat\rho|\psi\rangle|\psi\rangle}{\langle\psi|\psi\rangle^2}.
	\end{align}
	Identifying $g=\langle\psi|\hat\rho|\psi\rangle/\langle\psi|\psi\rangle$, this condition is identical to the eigenvalue problem of the density operator
	\begin{align}
		\hat\rho|\psi\rangle=g|\psi\rangle.
	\end{align}

	Thus, an orthonormal eigenbasis $\{|\psi_n\rangle|\}_n$ defines the set of stationary points $\mathcal D=\{|\psi_n\rangle\langle \psi_n|\}$.
	Consequently, we get for the definitions in Corollary \ref{Theo:RecCO} the Gram-Schmidt matrix $\boldsymbol G=[\mathrm{tr}(|\psi_n\rangle\langle\psi_n|\psi_m\rangle\langle\psi_m|)]_{m,n}=\mathrm{diag}(1,1,\ldots)$ and the vector of the nonnegative eigenvalues $\vec g=[\langle\psi_n|\hat\rho|\psi_n\rangle]_n=[g_n]_n$.
	Now the only solution of $\boldsymbol G\vec p=\vec g$ is obtained when $p_n=g_n$ for all $n$.

	Thus, the spectral decomposition \eqref{eq:SpectralDecomposition}, a cornerstone of many problems in mathematical physics, is just a special case of our general treatment in which $\mathcal C$ contains all pure states.

\subsection{Orthonormal qudit states}

	A frequently studied example of classical states is given by orthonormal basis vectors of a qudit system $\{|j\rangle\}_{j=0,\ldots,d-1}$.
	Then the set $\mathcal C=\{|j\rangle\langle j|:j=0,\ldots,d-1\}$ of classical states only includes isolated points.
	For example, a classical bit ($d=2$) is restricted to $|0\rangle\langle 0|$ and $|1\rangle\langle 1|$ corresponding to the classical truth values ``false'' and ``true,'' respectively.

	Since any isolated point is also stationary, we can identify $\mathcal D=\mathcal C$ in the scenario under study.
	Applying our technique, we get $\vec p=(p_0,\ldots,p_{d-1})^\mathrm{T}=[\langle j|\hat\rho|j\rangle]_{j=0,\ldots,d-1}$, which is a nonnegative vector.
	However, we observe that the residual component
	\begin{align}
		\hat \rho_\mathrm{res.}=\hat\rho-\sum_{j=0}^{d-1} \langle j|\hat\rho|j\rangle |j\rangle\langle j|
	\end{align}
	vanishes if and only if $\hat\rho$ is diagonal in $|j\rangle\langle j|$.
	In other words, the classical states take the form $\hat\rho=\sum_{j=0}^{d-1} p_j |j\rangle\langle j|$.
	This is certainly not surprising.
	However, again, it demonstrates the general function of our method.

	As an example, let us assume the state $\hat\rho=|\phi\rangle\langle\phi|$, using the coherent superposition $|\phi\rangle=(|0\rangle+\cdots+|d-1\rangle)/\sqrt d$.
	From Eq. \eqref{eq:RC} we obtain a residual component with
	\begin{align}
		\|\hat\rho_\mathrm{res.}\|^2=\mathrm{tr}\left[\left(|\phi\rangle\langle\phi|-\frac{\hat 1}{d}\right)^2\right]=1-\frac{1}{d}>0,
	\end{align}
	which confirms the quantum nature.
	This residual component yields the Hilbert-Schmidt distance of the state to the set of classical ones and can be used as a quantifier of quantum coherence for the considered notion of classicality \cite{SAP17}.

	The findings in this and the previous example might not be surprising.
	However, we made the claim that out method universally applies to discrete-variable systems.
	Thus, we find that it is important to confirm that known results are retrieved.
	With the first two examples, we successfully demonstrate that this can be easily achieved with our approach.

\subsection{True, false, and undecidable}

	As a more sophisticated example, let us assume that the classical states in a qubit system are given by
	\begin{align}
	\begin{aligned}
		\mathcal C=\{|0\rangle\langle 0|,|1\rangle\langle 1\}\cup\left\{|\varphi\rangle\langle\varphi|:0\leq \varphi<2\pi\right\},
		\\
		\text{where }|\varphi\rangle=\frac{|0\rangle+e^{i\varphi}|1\rangle}{\sqrt 2}.
	\end{aligned}
	\end{align}
	This can be compared to a classical ternary logic, which consists of the isolated states ``false'' ($|0\rangle$) and ``true'' ($|1\rangle$) and is extended by including ``undecidable'' states.
	The latter states $|\varphi\rangle$ have an equal chance of being true or false and form a continuum.
	In the Bloch-sphere representation, the convex set of classical states defines a double cone structure (see Fig. \ref{fig:Qubit}).
	For realistic problems in classical computation, the introduction of undecidable states is one way, for example, to allow for the possibility to indicate scenarios in which insufficient data are available to make a definitive decision.

\begin{figure}
	\includegraphics[width=\columnwidth]{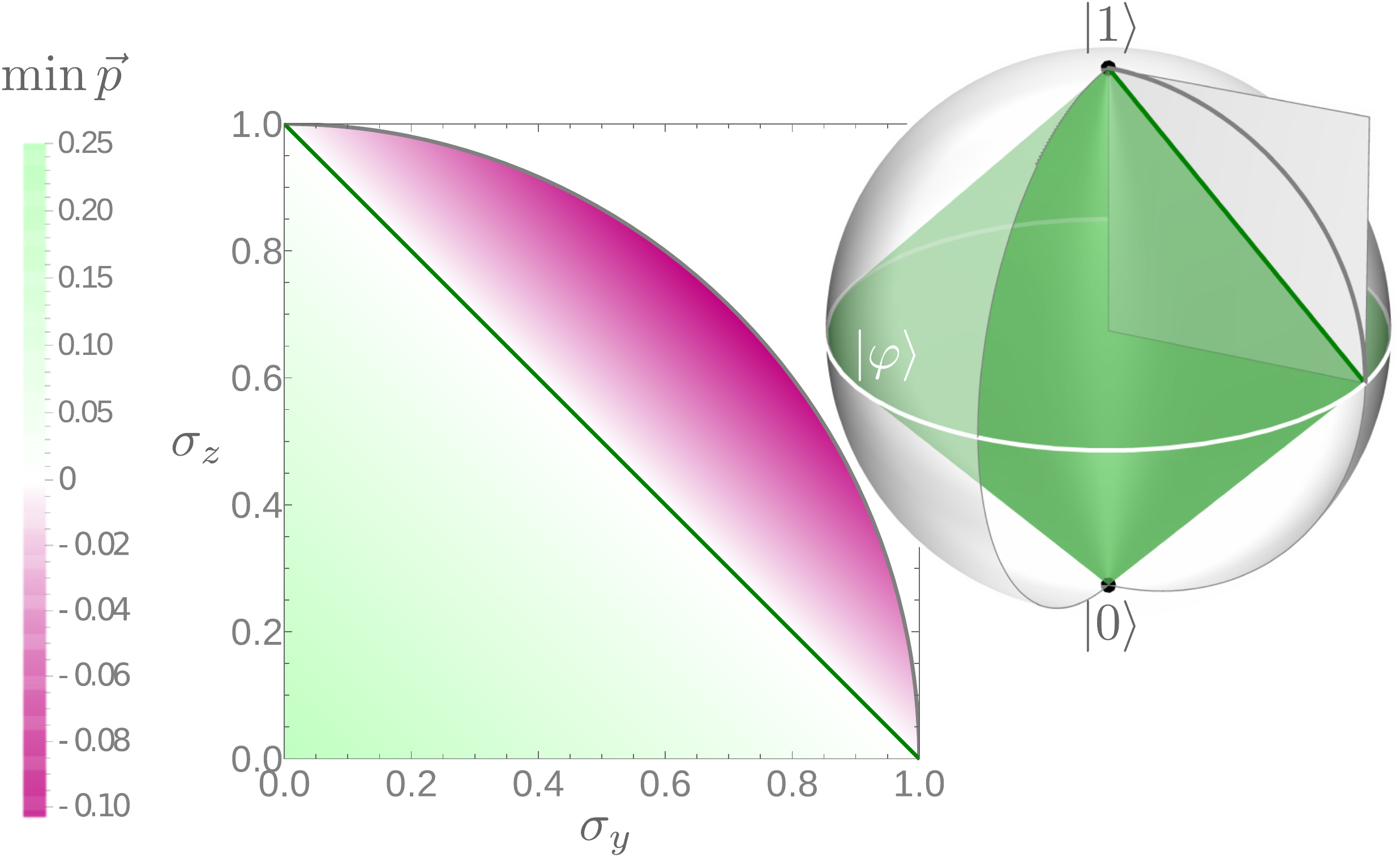}
	\caption{
		The green volume within the Bloch sphere (right) determines all classical states.
		The left plot depicts the minimal element of the quasiprobability $\vec p$ [cf. Eqs. \eqref{eq:QubitSolution1} and \eqref{eq:QubitSolution2}] for any cross section of the Bloch sphere, because of symmetry.
		The states in the lower triangle are decomposed in terms of a non-negative $\vec p$, $\min\vec p\geq 0$.
		The remaining states are identified as nonclassical, $\min\vec p<0$ (magenta).
	}\label{fig:Qubit}
\end{figure}

	A general qubit state can be decomposed as
	\begin{align}
	\begin{aligned}
		\hat\rho
		=&p|0\rangle\langle 0|+(1-p)|1\rangle\langle 1|
		\\
		&+\sqrt{p(1-p)}\Big(
			\gamma |1\rangle\langle 0|
			+\gamma^\ast |1\rangle\langle 0|
		\Big),
	\end{aligned}
	\end{align}
	with $|\gamma|\leq 1$ and $0\leq p\leq 1$.
	The stationary states are the isolated points $|0\rangle$ and $|1\rangle$ as well as the contributions $(|0\rangle+ e^{i\arg\gamma}|1\rangle)/\sqrt 2$ and $(|0\rangle- e^{i\arg\gamma}|1\rangle)/\sqrt 2$ for $|\gamma|\neq0$ of the continuous part obtained from solving $\partial_\varphi\langle\varphi|\hat\rho|\varphi\rangle=0$ [Eq. \eqref{eq:OptProbl}].
	Applying Corollary \ref{Theo:RecCO} and using the order $\{|0\rangle,|1\rangle,|\arg \gamma\rangle, |\arg \gamma+\pi\rangle\}$ for the stationary points, we get the matrix
	\begin{align}
		G=\begin{pmatrix}
			1 & 0 & 1/2 & 1/2 \\
			0 & 1 & 1/2 & 1/2 \\
			1/2 & 1/2 & 1 & 0 \\
			1/2 & 1/2 & 0 & 1
		\end{pmatrix},
	\end{align}
	which can be characterized through its eigenvectors $\vec e_\mathrm{norm.}=(1,1,1,1)^\mathrm{T}$, $\vec e_\mathrm{iso.}=(1,-1,0,0)^\mathrm{T}$, $\vec e_\mathrm{cont.}=(0,0,1,-1)^\mathrm{T}$, and $\vec e_\mathrm{ker.}=(1,1,-1,-1)^\mathrm{T}$ to the corresponding eigenvalues $2$, $1$, $1$, and $0$, respectively.
	Further, we obtain from Corollary \ref{Theo:RecCO} the vector
	\begin{align}
		&\vec g=\begin{pmatrix}
			p \\ 1-p \\ 1/2+\sqrt{p(1-p)}|\gamma| \\ 1/2-\sqrt{p(1-p)}|\gamma|
		\end{pmatrix}
		\\\nonumber
		=&\frac{1}{2}\vec e_\mathrm{norm.}+\frac{2p-1}{2}\vec e_\mathrm{iso.}+\sqrt{p(1-p)}|\gamma|\vec e_\mathrm{cont.}.
	\end{align}

	Solving the linear equation $G\vec p=\vec g$, we find
	\begin{align}
		\label{eq:QubitSolution1}
		&\vec p=\begin{pmatrix}
			p+r \\ 1-p+r \\
			\sqrt{p(1-p)}|\gamma|-r \\ -\sqrt{p(1-p)}|\gamma|-r
		\end{pmatrix}
		\\\nonumber
		=&\frac{1}{4}\vec e_\mathrm{norm.}{+}\frac{2p{-}1}{2}\vec e_\mathrm{iso.}
		{+}\sqrt{p(1{-}p)}|\gamma|\vec e_\mathrm{cont.}{+}\left[r+\frac{1}{4}\right] \vec e_\mathrm{ker.}
	\end{align}
	for any $r\in\mathbb R$.
	Note that $\hat\rho_\mathrm{res.}=0$ hods true as a consequence of the fact that the classical states have a nonempty volume in the Bloch sphere of all states.
	The found solution implies that a nonnegative $\vec p$ exists if and only if we can choose $r$ such that $\min\{p,1-p\}\geq -r\geq \sqrt{p(1-p)}|\gamma|$.
	Thus, for instance, we can set
	\begin{align}
		\label{eq:QubitSolution2}
		r=-\frac{\min\{p,1-p\}+\sqrt{p(1-p)}|\gamma|}{2},
	\end{align}
	which yields a nonnegative $\vec p$ if possible.
	Therefore, we obtain the nonclassicality of all qubit states in the considered scenario, which is shown in Fig. \ref{fig:Qubit} and is directly achieved by applying our method.

\subsection{Angular momentum coherent states}\label{Subsec:1Spin}

	As an example motivated by physics, a spin-$s$ system is considered with a total spin $s\in\mathbb N/2$.
	If $s$ is a half-integer or integer, we have fermionic or bosonic quantum characteristics, respectively, which is of particular interest when studying quantum correlations (cf. Sec. \ref{Subsec:2Spin}).
	The operator components of the spin satisfy $[\hat S_x,\hat S_y]=i\hbar\hat S_z$ and cyclic permutations thereof.
	The spectral decomposition, for example, of the $z$ component reads $\hat S_z=\hbar\sum_{m=-s}^s m|m\rangle\langle m|$.
	Using ladder operators, $\hat S_\pm |m\rangle=\hbar\sqrt{(s\mp m)(s\pm m+1)}|m\pm 1\rangle$, the $x$ and $y$ components can be written as $\hat S_x=(\hat S_++\hat S_-)/2$ and $\hat S_y=(\hat S_+-\hat S_-)/(2i)$.
	Furthermore, the classical spin states are defined as
	\begin{align}
	\label{eq:SpinClass}
	\begin{aligned}
		|\vartheta,\varphi\rangle
		=\sum_{m=-s}^s&\binom{2s}{s+m}^{\frac{1}{2}}
		\left[\cos\frac{\vartheta}{2}\right]^{s+m}
		\left[\sin\frac{\vartheta}{2}\right]^{s-m}
		e^{-im \varphi}
		|m\rangle,
		\\
	\end{aligned}
	\end{align}
	for $0\leq\vartheta\leq\pi$ and $0\leq\varphi<2\pi$.
	In addition, we have $\langle\vartheta,\varphi|\hat S_z|\vartheta,\varphi\rangle=\hbar s\cos\vartheta$ and $\langle\vartheta,\varphi|\hat S_\pm|\vartheta,\varphi\rangle=\hbar s\, e^{\pm i\varphi} \sin\vartheta$.
	Note that for $\vartheta\in\{0,\pi\}$, the angle $\varphi$ becomes irrelevant.
	Originally such states were introduced and studied as angular momentum coherent states \cite{AD71,L84} for generalizing quantum-optical concepts to angular-momentum-based algebras.
	A convex mixture of classical spin states then defines the notion of an incoherent states in such a scenario \cite{GBB08}.

	To apply our technique, we consider the pure superposition state $\hat\rho=|\psi\rangle\langle \psi|$, with
	\begin{align}
		|\psi\rangle=\frac{|-s\rangle+|s\rangle}{\sqrt 2},
	\end{align}
	where $s\geq 1$, to exclude trivial cases.
	The optimization of $\langle\vartheta,\varphi|\hat\rho|\vartheta,\varphi\rangle$ over $\vartheta$ and $\varphi$ can be done straightforwardly.
	This yields the solutions $\vartheta\in\{0,\pi\}$, resulting in the states $|{\pm}s\rangle$, as well as $\vartheta=\pi/2$ and $\varphi_n=\pi n/(2s)$ for $n=0,\ldots,4s-1$, which corresponds to the states
	\begin{align}
		|\pi/2,\varphi_n\rangle
		=&\frac{1}{2^s}
		\sum_{m=-s}^{s}\sqrt{\frac{(2s!)}{(s+m)!(s-m)!}} e^{-i\pi m n/(2s)}|m\rangle.
	\end{align}
	After some algebra, one obtains the quasiprobabilities
	\begin{align}
		\label{eq:spinQP}
		p_{\vartheta=0}=p_{\vartheta=\pi}=1/2
		\text{ and }
		p_{\vartheta=\pi/2,\varphi_n}=(-1)^n\frac{2^{2s}}{8s},
	\end{align}
	where the latter are negative for odd integers $n$.
	Note that we have a unique decomposition as $G$ is invertible and we have $\hat\varrho_\mathrm{res.}=0$.

	In Fig. \ref{fig:spin} the quasiprobabilities are shown over the spin-$s$ phase space, which is defined by the Poincar\'e sphere and must not be confused with the Bloch sphere for qubits.
	The quasiprobability (radial component) for the classical states in Eq. \eqref{eq:SpinClass} is shown as a function of the angles $\vartheta$ and $\varphi$.
	The top (bottom) plot shows the quasiprobability of the state $|\psi\rangle\langle\psi|$ under study for a boson (fermion).
	The negativities in the equatorial plane confirm the nonclassical character [Eq. \eqref{eq:spinQP}].

	Again, we are able to directly apply our technique to unambiguously confirm the quantum characteristic of states in this physical system, which is only remotely related to the previous example inspired by quantum information.

\begin{figure}
	\includegraphics[width=\columnwidth]{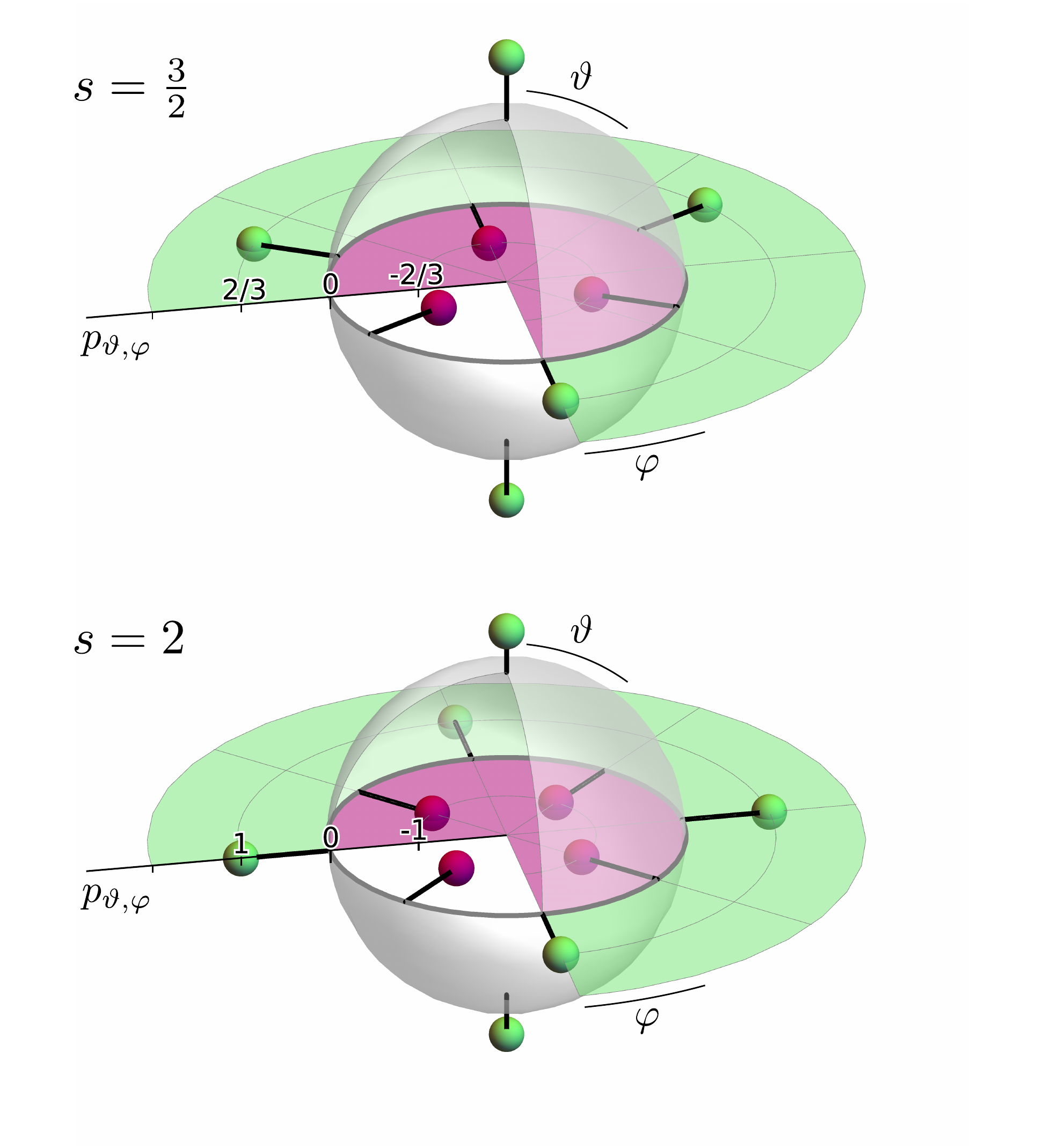}
	\caption{
		Quasiprobabilities over the Poincar\'e-sphere phase space.
		The polar angel is $\vartheta$ and the azimuthal angle is $\varphi$.
		Both define the classical states $|\vartheta,\varphi\rangle\langle\vartheta,\varphi|$ [Eq. \eqref{eq:SpinClass}].
		The radius over the sphere is used to show the quasiprobability $p_{\vartheta,\varphi}$ (bullet points) of the corresponding state; negative and positive values point inside (magenta) and outside (green), respectively.
		The negativities (magenta bullets) confirm the nonclassicality of the state $|\psi\rangle=(|s\rangle+|-s\rangle)/\sqrt2$.
	}\label{fig:spin}
\end{figure}

\section{Application to correlations}\label{Sec:Corr}

\subsection{Bipartite entanglement}\label{Subsec:BipartEnt}

	So far, we focused on the nonclassical characteristics of a single system.
	When studying quantum correlations between two systems, the set $\mathcal C$ of pure classically correlated states is formed by tensor-product states $|a,b\rangle=|a\rangle\otimes|b\rangle$ \cite{SAWV17}.
	Let us assume that all local states $|a\rangle$ and $|b\rangle$ of the subsystems $A$ and $B$ are classical.
	Then statistical mixtures of such states [Eq. \eqref{eq:IncoherentDecomposition}] define classically correlated states, termed separable states \cite{W89}.
	Thus, a nonclassical state is entangled in this case (see Ref. \cite{HHHH09} for an introduction to quantum entanglement).

	Similarly to the example for the spectral decomposition (Sec. \ref{Subsec:Spectral}), a map $\hat\Gamma(t)=(|a\rangle\langle a|/\langle a|a\rangle)\otimes(|b\rangle\langle b|/\langle b|b\rangle)$ is defined through the parameter $t=(\langle a|,\langle b|)$.
	In close analogy to the eigenvalue problem of the density operator, the optimization of $(\hat \Gamma(t)|\hat\rho)$ over the two components of $t$ results in the so-called separability eigenvalue equations,
	\begin{align}
		\label{eq:SEE}
		\hat\rho_b|a\rangle=g|a\rangle
		\text{ and }
		\hat\rho_a|b\rangle=g|b\rangle,
	\end{align}
	with the reduced operators $\hat\rho_a=(\langle a|\otimes\hat 1_B)\hat\rho(|a\rangle\otimes\hat 1_B)$ and $\hat\rho_b=(\hat 1_A\otimes \langle b|)\hat\rho(1_A\otimes|b\rangle)$.
	The thorough derivation of those equations was done in Ref. \cite{SV09}, where this approach was used to construct entanglement witnesses.

	It is well known that any quantum state can be written as a quasi-mixture of separable states (see, e.g., Ref. \cite{STV98}).
	Thus, one necessarily obtains $\hat\rho_\mathrm{res.}=0$.
	We previously constructed quasiprobabilities for bipartite entanglement based on Eq. \eqref{eq:SEE} \cite{SV09quasi}.
	Again, our method derived here provides a general framework, which includes such previous results as a special case.
	Existing examples of studied states are quantum-optical two-mode squeezed states \cite{SV12}, NOON states in quantum metrology \cite{BSV17}, and Werner states \cite{TBV17}.
	It is also noteworthy that a generalization of Eq. \eqref{eq:SEE} allows us to identify entanglement between quantum trajectories \cite{SW17}.

	In the following, we study further examples.
	Moreover, the previous results have been limited to bipartite entanglement.
	Based on our general approach, we are now able to address multipartite systems too (Sec. \ref{Subsec:MultipartEnt}).

\subsection{Complex probability amplitudes}\label{Subsec:Redits}

	To study combinations of local and nonlocal notions of nonclassicality let us consider the concept of a ``redit,'' real qudit.
	In this context, a pure state is classical if it can be written as real-valued vector in the computational basis.
	Thus, the set $\mathcal C$ of pure redit states $|c\rangle\langle c|$ is defined via 
	\begin{align}
		|c\rangle=\sum_{j=0}^{d-1}\gamma_j|j\rangle,
		\text{ where }\gamma_j\in\mathbb R
		\text{ for all }j.
	\end{align}
	The observation that there are states beyond redits is a consequence of the fact that quantum physics is described in terms of complex probability amplitudes $\gamma_j$ rather than real-valued vectors as one may expects from classical field theories.
	As we will visualize with our approach, interferences stemming from superpositions in a real vector space do not account for all phenomena which are described in complex spaces.
	Thus, the identification of correlations beyond redits confirms the fundamental axiom of quantum physics that a complex Hilbert space is required to describe quantum-physical states.

	In general, the Hermitian density operator can be decomposed into two real-valued matrices
	\begin{align}
		\hat\rho=\hat\rho_\mathrm{Re}+i\hat\rho_\mathrm{Im},
	\end{align}
	where $\hat\rho_\mathrm{Re}=(\hat\rho+\hat\rho^\ast)/2=\hat\rho_\mathrm{Re}^\mathrm{T}$ and $\hat\rho_\mathrm{Im}=(\hat\rho+\hat\rho^\ast)/(2i)=-\hat\rho_\mathrm{Im}^\mathrm{T}$.
	Further, we find for the classical states $\hat c=|c\rangle\langle c|$ that
	\begin{align}
		\mathrm{tr}(\hat c\hat\rho_\mathrm{Im})
		=\mathrm{tr}((\hat c\hat\rho_\mathrm{Im})^\mathrm{T})
		=-\mathrm{tr}[\hat\rho_\mathrm{Im}\hat c]
		=-\mathrm{tr}(\hat c\hat\rho_\mathrm{Im}).
	\end{align}
	This means that $(\hat c|\hat\rho_\mathrm{Im})=0$.
	Consequently, applying our approach to the single-mode scenario, we find that the stationary states are obtained via the eigenvalue problem of the real part of the density operator
	\begin{align}
		\label{eq:EEreal}
		\hat\rho_\mathrm{Re}|c\rangle=g|c\rangle.
	\end{align}
	This implies that the spectral decomposition over the real-valued vectors yields probabilities which are nonnegative.
	Still, nonclassicality can be inferred from the residual operator, $\hat\rho_\mathrm{res.}=\hat\rho-\hat\rho_\mathrm{Re}=i\hat\rho_\mathrm{Im}$.

	For studying correlations, we can combine the approaches in Eqs. \eqref{eq:SEE} and \eqref{eq:EEreal} to study the case of separable and real-valued states.
	As an example, we consider the two-qubit state
	\begin{align}
		\label{eq:2QubitFamily}
		\hat\rho=\frac{\hat \sigma_0\otimes\hat\sigma_0+\sum_{w\in\{z,x,y\}}\rho_w\hat\sigma_w\otimes\hat\sigma_w}{4},
	\end{align}
	where the $\hat \sigma_w$ denote Pauli matrices for $w\in\{x,y,z\}$ \cite{comment:Pauli}.
	This four-dimensional density operator is real-valued with the eigenvalues $(1+\rho_z+\rho_x-\rho_y)/4$, $(1+\rho_z-\rho_x+\rho_y)/4$, $(1-\rho_z+\rho_x+\rho_y)/4$, and $(1-\rho_z-\rho_x-\rho_y)/4$.
	Consequently, the parameter space of physical quantum states in Eq. \eqref{eq:2QubitFamily} is the tetrahedron spanned by $(\rho_x,\rho_y,\rho_z)\in\{(-1,-1,-1),(-1,1,1),(1,-1,1),(1,1,-1)\}$.

	The complex form of the separability eigenvalue equations \eqref{eq:SEE} for the density operator under study is solved by
	\begin{align}
		|a\rangle\otimes |b\rangle
		=
		|w_{\pm_A}\rangle\otimes|w_{\pm_B}\rangle,
	\end{align}
	for $w\in\{z,x,y\}$ and where $|w_\pm\rangle$ are the eigenvectors of the Pauli matrix $\hat\sigma_w$ \cite{comment:Pauli}.
	The signs $\pm_A$ and $\pm_B$ can be chosen arbitrarily for the subsystems.
	In the scenario of real-valued product vectors, i.e., solving the separability eigenvalue problem for real numbers, we obtain the same solutions with the restriction to $w\in\{z,x\}$.

	Consequently, the latter set of solutions allows us to obtain the probability vector for $\mathbb R$,
	\begin{align}
		\vec p=\begin{pmatrix}
			p_{z,+,+}\\p_{z,+,-}\\p_{z,-,+}\\p_{z,-,-}\\
			p_{x,+,+}\\p_{x,+,-}\\p_{x,-,+}\\p_{x,-,-}
		\end{pmatrix}
		=\frac{1}{8}\begin{pmatrix}
			1\\1\\1\\1\\1\\1\\1\\1
		\end{pmatrix}
		+\frac{1}{4}\begin{pmatrix}
			\rho_z\\-\rho_z\\-\rho_z\\\rho_z\\
			\rho_x\\-\rho_x\\-\rho_x\\\rho_x
		\end{pmatrix}
		+\begin{pmatrix}
			r\\r\\r\\r\\-r\\-r\\-r\\-r
		\end{pmatrix},
	\end{align}
	where the indices are arranged according to $w,\pm_A,\pm_B$ and a free parameter $r$.
	We find that a positive solution $\vec p$ is possible if and only if
	\begin{align}
		r\leq \frac{1}{8}-\frac{|\rho_x|}{4}
		\text{ and }
		r\geq -\frac{1}{8}+\frac{|\rho_z|}{4}
	\end{align}
	can be satisfied.
	For instance, the mean value of the upper and lower bound, $r=(|\rho_z|-|\rho_x|)/8$, ensures a positive $\vec p$ if possible.
	Moreover, the residual component for the real numbers is $\hat\rho_\mathrm{res.}=\rho_y\hat\sigma_y\otimes\hat\sigma_y/4$.

	For the case $\mathbb C$, we append the $y$ components and get
	\begin{align}
		\vec p
		=\frac{1}{12}\begin{pmatrix}
			1\\1\\1\\1\\1\\1\\1\\1\\1\\1\\1\\1
		\end{pmatrix}
		+\frac{1}{4}\begin{pmatrix}
			\rho_z\\-\rho_z\\-\rho_z\\\rho_z\\
			\rho_x\\-\rho_x\\-\rho_x\\\rho_x\\
			\rho_y\\-\rho_y\\-\rho_y\\\rho_y
		\end{pmatrix}
		+\begin{pmatrix}
			r+q\\r+q\\r+q\\r+q\\
			-r\\-r\\-r\\-r\\
			-q\\-q\\-q\\-q
		\end{pmatrix},
	\end{align}
	for arbitrary $r$ and $q$.
	Now a positive solutions exists if and only if the inequalities
	\begin{align}
	\begin{aligned}
		r\leq \frac{1}{12}-\frac{|\rho_x|}{4},
		\text{ }
		q\leq \frac{1}{12}-\frac{|\rho_y|}{4},
		\text{ and }
		r+q\geq -\frac{1}{12}+\frac{|\rho_z|}{4}
	\end{aligned}
	\end{align}
	simultaneously hold true.
	Similarly to the previous scenario, we can make the symmetric choice $r=(|\rho_z|+|\rho_y|-2|\rho_x|)/12$ and $q=(|\rho_z|+|\rho_x|-2|\rho_y|)/12$ to obtain $\vec p\geq 0$ if possible.
	Further, we have $\hat\rho_\mathrm{res.}=0$ for the case $\mathbb C$.

	A specific state of the form in Eq. \eqref{eq:2QubitFamily} was studied in Ref. \cite{CFR01} and it is an interesting example of so-called bound entanglement for a rebit.
	For this state, we have $\rho_x=\rho_z=0$ and $\rho_y=1$.
	Then we get in the real-valued case that $\vec p\geq 0$, but $\hat\rho_\mathrm{res.}\neq 0$ certifying the nonclassical character, i.e., real-valued entanglement.
	By contrast, the complex scenario yields a convex decomposition, verifying that the state is separable with respect to $\mathbb C$.
	Thus, our more general solutions directly confirm the predicted properties of this state \cite{CFR01}.

\begin{figure}
	\includegraphics[width=\columnwidth]{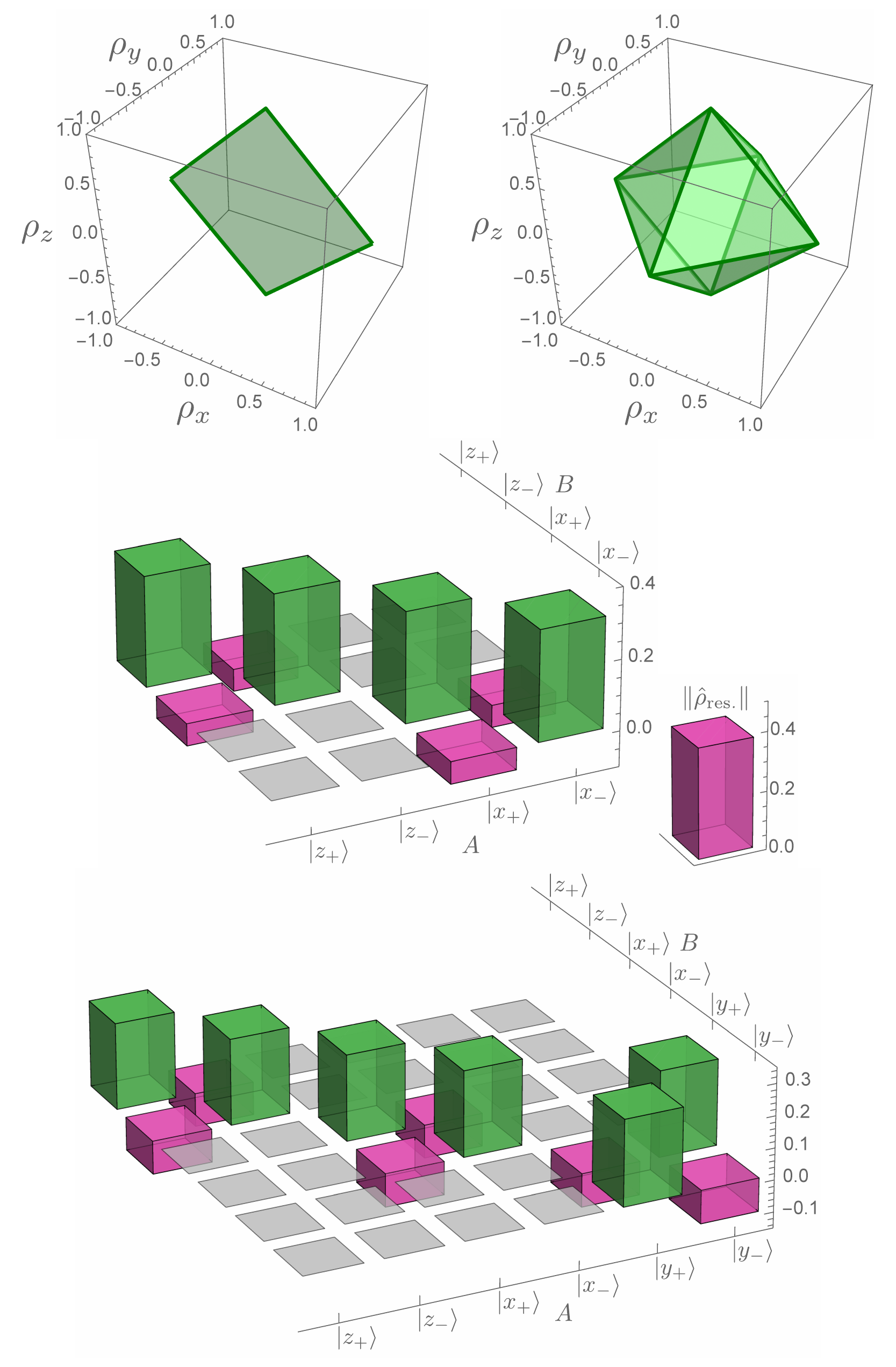}
	\caption{
		The top row shows the green regions depicting separable states over $\mathbb R$ (left) and $\mathbb C$ (right) for the family of states in Eq. \eqref{eq:2QubitFamily}.
		The middle row shows, for $\rho_x=-\rho_y=\rho_z=3/4$, the quasiprobabilities for $|a\rangle\otimes|b\rangle$ are shown---the $x$ and $y$ axes define the states in subsystems $A$ and $B$, respectively.
		The Hilbert-Schmidt norm of its residual component is depicted left for this rebit scenario.
		The bottom row shows, analogously, the qubit scenario ($\|\hat\rho_\mathrm{res.}\|=0$).
	}\label{fig:ReBit}
\end{figure}

	Beyond this specific example, we can consider the general case.
	Based on our approach, we can identify all pure and mixed classical states in the given parametrization [Eq. \eqref{eq:2QubitFamily}];
	the top row in Fig. \ref{fig:ReBit} shows the real-valued (complex-valued) separable parameter space on the left (right), which forms an octahedron (square).
	In addition, for the parameters $\rho_x=-\rho_y=\rho_z=3/4$, the quasiprobabilities are shown in the middle row of Fig. \ref{fig:ReBit} for the $\mathbb R$ case.
	The bottom row shows the $\mathbb C$ scenario in which the residual component of the previous case is resolved into the additional $y$ components.
	Note that the sign of $\rho_y$ is opposite to those of the $x$ and $z$ components, resulting in different quasiprobabilities.

	This study demonstrates that our approach renders it possible to perform an intuitive and joint characterization of entanglement together with the fact that quantum physics requires complex probability amplitudes in a unified framework.

\subsection{Bosons and fermions}\label{Subsec:2Spin}

	Apart from entanglement, there are other forms of quantum correlations emerging in composite systems.
	Previously, we studied spin states (Sec. \ref{Subsec:1Spin}).
	Here let us consider such states for a bipartite system together with the exchange symmetry for bosons and fermions.

	Again, we can apply our general approach to this scenario (cf. Corollaries \ref{Theo:ConvDec} and \ref{Theo:RecCO}).
	As an example, we study the state
	\begin{align}
		\label{eq:2particlespin}
		|\psi\rangle=\frac{1}{\sqrt 2}\left\{\begin{array}{ll}
			|s,-s\rangle+|-s,s\rangle & \text{ for integer $s$},
			\\
			|s,-s\rangle-|-s,s\rangle & \text{ for half-integer $s$},
		\end{array}\right.
	\end{align}
	where $s\geq1$.
	This state shares similarities with a NOON state, which is frequently applied in quantum metrology because of its advanced phase precision \cite{BKABWD00}.
	The optimization over the states $|\vartheta_A,\varphi_A\rangle\otimes|\vartheta_B,\varphi_B\rangle$ first results in the optimal states with polar angles $0$ and $\pi$, yielding $|s\rangle\otimes|s\rangle$, $|-s\rangle\otimes|s\rangle$, $|s\rangle\otimes|-s\rangle$, and $|-s\rangle\otimes|-s\rangle$.
	Second, we get the relation between the azimuthal angles
	\begin{align}
		\varphi_A=\frac{\pi n}{2s}+\varphi_B\text{ for }n\in\mathbb Z,
	\end{align}
	and the relation between the polar angles
	\begin{align}
		\vartheta_A=\left\lbrace\begin{array}{ll}
			\pi-\vartheta_B & \text{for $n+2s$ even},\\
			\vartheta_B & \text{for $n+2s$ odd}.
		\end{array}\right.
	\end{align}
	Thus, in analogy to the single system, we can chose $\vartheta_A=\vartheta_B=\pi/2$, $\varphi_A=\pi n_A/(2s)$ and $\varphi_B=\pi n_B/(2s)$ to get a finite and exhaustive subset of stationary points.
	Then we get
	\begin{align}
	\begin{aligned}
		p_{\vartheta_A=0,\vartheta_B=0}=p_{\vartheta_A=\pi,\vartheta_B=\pi}=&0,
		\\
		p_{\vartheta_A=0,\vartheta_B=\pi}=p_{\vartheta_A=\pi,\vartheta_B=0}=&\frac{1}{2},
		\\
		p_{\vartheta_A=\vartheta_B=\pi/2,\varphi_A=\pi n_A/(2s),\varphi_B=\pi n_B/(2s)}=&(-1)^{n_A+n_B+2s}\frac{2^{4s}}{64s},
	\end{aligned}
	\end{align}
	as well as $\hat\rho_\mathrm{res.}=0$

	The found quasiprobabilities are negative when $n_A+n_B+2s$ is an odd integer.
	This confirms that the state in Eq. \eqref{eq:2particlespin} is nonclassically correlated in the sense that it is not a convex combination of tensor products of classical spin states $|\vartheta_A,\varphi_A\rangle\otimes|\vartheta_B,\varphi_B\rangle$.
	Furthermore, it is obvious that we can write $|\psi\rangle\propto |0,\varphi_A\rangle\vee |\pi,\varphi_B\rangle$ for bosons (for integer $s$ and using the symmetric tensor product $\vee$) as well as $|\psi\rangle\propto |0,\varphi_A\rangle\wedge |\pi,\varphi_B\rangle$ for fermions (for half-integer $s$ and using the skew-symmetric tensor product $\wedge$).
	Thus, when including the exchange symmetry originating from the spin statistics, the state is indeed a classical symmetric or skew-symmetric tensor-product state.
	Consequently, we observe that the negativities of the quasiprobability are a result of the exchange symmetry in quantum physics.
	See also Ref. \cite{SPBW17} for a recent in-depth ana\-lysis of the interplay between local quantum coherence, entanglement, and the exchange symmetry of indistinguishable particles.

\subsection{Multipartite entanglement}\label{Subsec:MultipartEnt}

	In our study of quantum correlations, we focused on bipartite systems so far.
	In this last application, let us consider multipartite entanglement of a complex-valued $N$-partite system.
	The pure classical states are all $N$-fold tensor-product vectors $|\psi_1,\ldots,\psi_N\rangle=|\psi_1\rangle\otimes\cdots\otimes|\psi_N\rangle$.
	Similarly to the bipartite scenario, we obtain the stationary points of a state $\hat\rho$ in terms of a multipartite version of the separability eigenvalue equations \cite{SV13}
	\begin{align}
		\label{eq:NSeparabilityEigenvalue}
		\hat\rho_{\psi_1,\ldots,\psi_{j-1},\psi_{j+1},\ldots,\psi_N}|\psi_j\rangle=g|\psi_j\rangle
	\end{align}
	for $j=1,\ldots,N$ and the operators $\hat\rho_{\psi_1,\ldots,\psi_{j-1},\psi_{j+1},\ldots,\psi_N}=(\langle \psi_1,\ldots,\psi_{j-1}|\otimes\hat 1_j\otimes\langle \psi_{j+1},\ldots,\psi_N|)\hat\rho(|\psi_1,\ldots,\psi_{j-1}\rangle\otimes\hat 1_j\otimes|\psi_{j+1},\ldots,\psi_N\rangle)$, where $\hat 1_j$ is the identity of the $j$th subsystem.
	Since any operator can be expanded in terms of Hermitian product operators, the residual component has to vanish for any state, $\hat\rho_\mathrm{res.}=0$.
	Thus, our technique enables us to uncover multipartite entanglement of any state $\hat\rho$ in terms of quasiprobability distributions, $\vec p\ngeq 0$.

	An interesting example of a four-partite state was introduced by Smolin \cite{S01}, $\hat\rho=(\hat\sigma_0^{\otimes 4}+\hat\sigma_z^{\otimes 4}+\hat\sigma_x^{\otimes 4}+\hat\sigma_y^{\otimes 4})/16$.
	This state is bound entangled, which presents a ``weak'' form of quantum correlation.
	More generally, let us study the family of states \cite{AH06}
	\begin{align}
		\label{eq:GenSmolin}
		\hat\rho=\frac{1}{2^N}\left(\hat\sigma_0^{\otimes N}+\rho_z\hat\sigma_z^{\otimes N}+\rho_x\hat\sigma_x^{\otimes N}+\rho_y\hat\sigma_y^{\otimes N}\right),
	\end{align}
	for $N\geq 2$, which includes the Smolin state as a special case.

	The treatment of this class of state can be done analogously to the bipartite scenario considered previously (Sec. \ref{Subsec:Redits}).
	Note that we focus on complex-valued states here.
	We find the multipartite separability eigenvectors by solving Eq. \eqref{eq:NSeparabilityEigenvalue};
	they are \cite{comment:Pauli}
	\begin{align}
		\bigotimes_{j=1}^N(|\psi_j\rangle\langle\psi_j|)=\bigotimes_{j=1}^N\frac{\hat\sigma_0+s_j\hat\sigma_w}{2},
	\end{align}
	where $w\in\{z,x,y\}$ and for any choice of sign $s_j\in\{+1,-1\}$, resulting in $2^N$ combinations of signs for each $w$.
	Eventually, we obtain the quasiprobabilities as
	\begin{align}
	\begin{aligned}
		p_{z,s_1,\ldots,s_N}=&\frac{1}{3}\frac{1}{2^N}+s_1\cdots s_N\frac{\rho_z}{2^N}+(r+q),
		\\
		p_{x,s_1,\ldots,s_N}=&\frac{1}{3}\frac{1}{2^N}+s_1\cdots s_N\frac{\rho_x}{2^N}-r,
		\\
		p_{y,s_1,\ldots,s_N}=&\frac{1}{3}\frac{1}{2^N}+s_1\cdots s_N\frac{\rho_y}{2^N}-q.
	\end{aligned}
	\end{align}
	Setting $r=(-r'+1/3-|\rho_x|)/2^N$ and $q=(-q'+1/3-|\rho_y|)/2^N$, we get the constraints for classical probabilities, i.e., separability, $r'\geq0$, $q'\geq 0$, and $r'+q'\leq1-|\rho_z|-|\rho_x|-|\rho_y|$.
	For the upper bounds $r'=q'=0$, we get the necessary and sufficient condition that a generalized Smolin state in Eq. \eqref{eq:GenSmolin} is separable ($\vec p\geq0$) for any number of parties $N$,
	\begin{align}
		|\rho_z|+|\rho_x|+|\rho_y|\leq 1.
	\end{align}

	Conversely, the violation of this condition yields $\vec p\ngeq0$ and certifies multipartite entanglement.
	For instance, assume we mix the actual Smolin state with white noise $\eta\sum_{w\in\{0,z,x,y\}}\hat\sigma_w^{\otimes 4}/16+(1-\eta)\hat\sigma_0^{\otimes 4}/16$ (see the experimental realizations in Refs. \cite{AB09,BSGMCRHB10,LKPR10}).
	Then the above condition implies that as long as $\eta>1/3$ holds true, the state exhibits multipartite entanglement.

\section{Conclusion}\label{Sec:Conclusion}

	In summary, we developed a universal framework for the representation of quantum coherences in terms of quasiprobabilities.
	We proved that our construction yields a classical probability distribution for the convex expansion of any statistical mixture of pure classical states.
	This is important as the decomposition of quantum states is, in general, not unique and finding one negative distribution does not imply that no nonnegative distribution exists.
	Our approach overcomes this challenge even if we have a continuum of pure classical states.
	In particular, we have proven that the pure classical states with an optimal distance to the state under study enable a convex decomposition if possible.
	Thus, a state exhibits quantum coherences, i.e., it is nonclassical, if our reconstructed distribution includes negativities or does not expand the state under study, leading to a nonzero residual component not accessible with classical states.
	Therefore, our necessary and sufficient method provides an intuitive picture for certifying quantum phenomena and provides an optimal decomposition for mixed classical states in finite-dimensional spaces.
	Let us also mention that the direct generalization of our approach to any continuous-variable system is a nontrivial problem as the distributions can become highly irregular in some cases \cite{S16}.
	In addition to the derivation of our method, we applied our method to study various systems which are relevant for quantum physics and quantum information science.

	To demonstrate the general operation of our technique, we showed that the spectral decomposition of the density operator is just a special case of our general treatment when assuming that all pure states are classical.
	Then we analyzed a qudit system in which the classical states form an orthonormal basis.
	In this case, the computed distribution is always nonnegative and quantum coherence manifests itself exclusively via a non-vanishing residual component, which can be further used to quantify the nonclassicality.
	Moreover, for a single qubit, we studied the scenario where the classical states include not only the orthogonal states for true and false, but also their balanced superposition to represent a third truth value, termed undecidable.
	We computed the quasiprobabilities for all states in and on the Bloch sphere and found that nonclassicality exclusively stems from negativities in the optimal distribution as the residual component vanishes.
	Finally, we considered the phase-space representation of a quantum superposition state in a spin-$s$ system.
	This yields a quasiprobability distribution which directly visualizes the nonclassical character on the Poincar\'e sphere, the phase space of spin states.
	In fact, the so-called Holstein-Primakoff approximation \cite{HP40} for a large total spin yields a tangential plane to the Poincar\'e sphere which corresponds to the phase space of a continuous-variable harmonic oscillator used in quantum optics.

	Beyond a single system, we additionally investigated quantum correlations in composite systems.
	As a first example, we considered entanglement quasiprobabilities for bipartite entanglement.
	Specifically, we analyzed the impact of complex Hilbert spaces when compared to states over real numbers.
	With this example, we were able to jointly characterize quantum features stemming from entanglement and the necessity of complex-valued probability amplitudes in quantum physics.
	Furthermore, we considered quantum correlations originating from the exchange symmetry in systems of indistinguishable particles.
	Specifically, the bipartite quasiprobability distribution of a superposition of two fermions and bosons was determined.
	Eventually, we introduced quasiprobability distributions for multipartite entanglement.
	Here, we characterized the multipartite quantum correlations of an entire family of $N$-partite generalized Smolin states.

	In conclusion, a methodological framework is devised which provides a versatile approach to uncover various forms of quantum coherences and quantum correlations.
	As our decomposition strategy yields a convex expansion for any classical density operator, we are able to study the geometry of the set, or a subset, of mixed classical states, as we demonstrated for several examples.
	Furthermore, our technique is directly accessible and comparably simple, which is underlined by the fact that all results presented here are based on exact solutions.
	Moreover, our method is not limited to one specific form of quantum coherence in a quantum system.
	Therefore, we provide an easily accessible and intuitive approach to uncover general forms of quantum coherence via quasiprobabilities for characterizing quantum phenomena.

\begin{acknowledgments}
	This work has received funding from the European Union's Horizon 2020 Research and Innovation Program under grant agreement No. 665148 (QCUMbER).
\end{acknowledgments}



\begin{thebibliography}{99}
\bibitem{comment:Pauli} 
		The Pauli matrices are $\hat\sigma_x=|1\rangle\langle0|+|0\rangle\langle 1|$, $\hat\sigma_y=i|0\rangle\langle1|-i|1\rangle\langle 0|$, and $\hat\sigma_z=|1\rangle\langle1|-|0\rangle\langle 0|$.
		The identity reads $\hat\sigma_0=|0\rangle\langle 0|+|1\rangle\langle 1|$.
		The eigenvectors for the $x$, $y$, and $z$ component are $|x_\pm\rangle=(|1\rangle\pm|0\rangle)/\sqrt 2$, $|y_\pm\rangle=(|1\rangle\pm i|0\rangle)/\sqrt 2$, $|z_+\rangle=|1\rangle$, and $|z_-\rangle=|0\rangle$ to the eigenvalues $\pm 1$.
		Furthermore, one can write $|w_\pm\rangle\langle w_\pm|=(\hat\sigma_0\pm\hat\sigma_w)/2$ for $w=x,y,z$.
\end{thebibliography}
\end{document}